\documentclass[aps,pra,showpacs,twoside,twocolumn,10pt,nofootinbib, longbibliography]{revtex4-2}
\usepackage[colorlinks=true, citecolor=red, urlcolor=blue ]{hyperref}
\usepackage{times,epsfig,amssymb,amsfonts,amsmath,bm,subfigure,mathtools,amsthm,braket,soul,enumitem,color,physics,graphics,graphicx}
\usepackage[normalem]{ulem}
\usepackage{comment}
\usepackage{epstopdf}

\newtheorem{proposition}{Proposition}

\usepackage{optidef}

\begin{document}

\title{Disparity between multipartite entangling and disentangling powers of unitaries: Even vs Odd
}

\author{Mrinmoy Samanta, Sudipta Mondal, Aditi Sen (De)}
\affiliation{Harish-Chandra Research Institute, A CI of Homi Bhabha National Institute, Chhatnag Road, Jhunsi, Prayagraj 211 019, India}

\begin{abstract}
\label{sec:abs}

We compare the multipartite entangling and disentangling powers of unitary operators by assessing their ability to generate or eliminate genuine multipartite entanglement.  Our findings reveal that while diagonal unitary operators can exhibit equal entangling and disentangling powers, certain non-diagonal unitaries demonstrate an imbalance when acting on fully separable states, thereby extending the known disparity from bipartite systems to those with any number of parties.  Counterintuitively, we construct classes of unitaries and their adjoints  that display unequal entanglement generation capacities, behaving  differently when applied to systems with an even number of qubits compared to those with an odd number. Further, we illustrate that this asymmetry can be simulated using physically realizable Hamiltonians: systems with an even number of qubits employ nearest-neighbor Dzyaloshinskii–Moriya (DM) interactions, while those with an odd number utilize a combination of Heisenberg and DM interactions.
Additionally, we present a circuit composed of random noncommuting unitaries, constructed from alternating layers of two-qubit Haar-random gates, to illustrate the discrepancy in the entangling and disentangling capabilities of unitaries.


\end{abstract}
\maketitle
\section{Introduction}
\label{sec:intro}

Over the past few decades, quantum algorithms \cite{Nielsen2000QuantumComputation} have revolutionized computational complexity, surpassing existing classical algorithms. Notable breakthroughs include the Deutsch–Jozsa and Simon’s algorithms, which demonstrate quantum speedup in oracle problems \cite{Deutsch_1992, Simon1997}, Shor’s algorithm to factor integers into their prime factors, providing  exponential speedup  \cite{Shor1997} and Grover’s  algorithm, offering quadratic improvement for unstructured search tasks \cite{Grover1996}. 
 Moreover, recent experimental achievements, such as Google's demonstration of quantum supremacy through random circuit sampling, have outperformed classical supercomputers~\cite{Arute2019}. Additionally, algorithms like Deutsch–Jozsa, Grover’s, Shor's, and the Harrow-Hassidim-Lloyd algorithm for solving linear systems have been implemented across various physical platforms, including superconducting qubits, trapped ions, and photonic architectures~\cite{Brickman2005_pra, Walther2005, Lanyon2007_prl, Lu2007, Politi2009, Lucero2012, Cai2013_prl, Krantz2019, Kjaergaard2020, Takita_2017, Martin2023, trappedionexp}.  Note that the successful execution of quantum algorithms heavily relies  on the design of optimal quantum circuits composed of several quantum gates \cite{Barenco1995, KHANEJA200111, Barenco1995, Duan2019_pra, Guerreschi2022, Ruiz2025}.  These developments naturally lead to inquiries about the quantum resources that enable such computational advantages.
Indeed, it has been demonstrated that quantum benefits arise from the generation of quantum correlations, such as  entanglement~\cite{Horodecki_2009_RMP},  quantum discord  \cite{Ollivier_2001_prl,Modi2012,Bera_2017},  and  coherence \cite{Baumgratz_2014_prl,Radhakrishnan_2016}, having no classical analogues.  These correlations are harnessed through various operations within quantum circuits, underscoring their pivotal role in quantum computing \cite{Ekert1998, Azuma2001, Shimoni2005, kendon2006, Datta2008, Rossi2013, Shubhalakshmi2020, Naseri2022, maiti2022quantumphaseestimationpresence, Kumar2024, Agrawal2025}. 

To design efficient quantum circuits, it is essential to understand how quantum gates or unitary operators generate or destroy quantum resources, especially entanglement, required for quantum information processing \cite{Zanardi00, Zanardi01, Zanardi02, Zanardi04, Linden2009_prl}. 
In pursuit of this goal, the entangling power of a unitary operator, \(U\), was introduced, which measures its capacity to generate entanglement from product states \cite{Zanardi01}. On the other hand, the disentangling power, or entangling power of \(U^{\dagger}\), measures its ability to reduce entanglement when acting on entangled inputs \cite{Berry2003_pra, Zanardi00,  Zanardi01, Zanardi02, Zanardi04, Clarisse_2007, Linden2009_prl, balakrishnan_2010,galve_2013,Chen2016,Chen2016_pra,Shen_2018}.  Despite being mathematically reversible~\cite{Nielsen2000QuantumComputation}, unitary operations, \(U\), and their adjoints \( U^\dagger \) frequently have inherent asymmetry in their physical roles in quantum information processing, especially in the creation of entanglement, as shown for bipartite settings with different local dimensions ~\cite{Linden2009_prl} (when both parties have the same dimension, these two powers are equal~\cite{Berry2003_pra}).


With the rapid advancement of establishing scalability in quantum computing ~\cite{DiVincenzo2000,Fowler2012SurfaceCode,Campbell2017FTQC,Preskill2018NISQ,Arute2019Supremacy,Henriet2020NeutralAtoms} and quantum communication systems~\cite{kimble2008quantum,Cuquet2012,Wehner2018, aditi2010, Ben-Or1988MultipartyComputation,Kitaev1995QuantumCryptography,Shor1999QuantumSecretSharing,Cleve1999Fairplay,Gisin2002QuantumCryptography,Vazirani2003MultipartyQuantum,Gottesman2004QuantumProtocols,Broadbent2009QuantumMultiparty}, it is fascinating to determine if the disparity between entangling and disentangling powers persists for unitaries acting on an arbitrary number of parties.   
In this respect, although a limited number of studies introduce the notion of multipartite entangling power  \cite{Linowski2020, hazra2023hierarchiesgenuinemultipartiteentangling,mondal2023imperfectentanglingpowerquantum, Qiu2025}, a comprehensive study of the power of entanglement versus disentanglement in a multipartite setting is still lacking. In this work,  our aim is to fill this gap, as it is essential for entanglement generation procedures that are resistant to irreversibility on both a structural and an operational level. 

At the outset, we prove that {\it multipartite} entangling and disentangling capabilities are equivalent for a large class of diagonal unitary operators after their actions on a fully separable state, thereby establishing a form of reversibility. This equivalence is  further numerically confirmed for diagonal unitaries generated at random. However, this symmetry breaks down for non-diagonal unitaries. In particular, we construct a family of non-diagonal unitaries that exhibit a clear disparity between their entangling and disentangling capabilities. Extending this investigation, we exhibit that such irreversibility persists for Haar-random unitaries when the optimization is performed across  the set of fully separable inputs, highlighting the inherent inequivalence in generic non-diagonal unitary evolutions. Surprisingly, the unitaries revealing unequal entangling and disentangling powers differ between systems with an even and odd number of parties. Moreover,  we illustrate that for a lattice with an even number of qubits, unitaries with different entangling and disentangling powers can be simulated using engineered Hamiltonians involving the Dzyaloshinskii-Moriya (DM) interaction \cite{Dzyaloshinsky1958, Moriya1960, Moriya1960_prl, Jafari2008}, whereas for an odd number of sites, the governing Hamiltonian must include both Heisenberg and DM interactions.  Additionally, we provide two-layered quantum circuits \cite{Nahum2017_prx, Nahum2018_prx, Bao2020_prb, Bera2020, skinner2023lecturenotesintroductionrandom} comprised of random two-qubit quantum gates  that can also detect the difference between the odd- and even-number of inputs according to the entanglement of the outputs. 


Our paper is organized as follows. In Sec.~\ref{sec:definition}, we set the definitions of multipartite entangling and disentangling powers and we prove these quantities to be equal for multipartite diagonal unitary operators (see Sec.~\ref{sec:diagonal}). In Sec.~\ref{sec:nd_realization}, we explicitly demonstrate the difference between   multipartite entangling and disentangling powers for certain specifically constructed non-diagonal unitary operators. We show in Sec. \ref{sec:unitaryHamil} that this dissimilarity can also be generated through the dynamics via  Dzyaloshinskii–Moriya interacting Hamiltonian for an even number of qubits as well as through Heisenberg and DM interactions together for odd lattice sites. Finally, we conclude in Sec.~\ref{sec:conclusio}.

\section{Multipartite Entangling and Disentangling Powers of Unitary Operators}
\label{sec:definition}

Let us begin by defining the multipartite entangling and disentangling powers based on the concept of genuine multipartite entanglement (GME) measure (see Appendix \ref{sec:ggm}), namely generalized geometric measure (GGM) \cite{aditi2010, Wei2003_pra, Blasone2008_pra}.

\subsection{ Definition of multipartite entangling and disentangling powers}
\label{sec:def}

Given a unitary operator, its {\it multipartite} entangling power is determined by its ability to generate maximum genuine multipartite entanglement from a fully separable (FS) state. Quantitatively,
for a given GME measure, \( \mathcal{G} \), the multipartite entangling power of a unitary \( U \), acted on \(N\)-party states is defined as \cite{hazra2023hierarchiesgenuinemultipartiteentangling, mondal2023imperfectentanglingpowerquantum, Qiu2025}
\begin{eqnarray}
     \mathcal{E}_N(U) &=& \max_{\ket{\psi}_{\text{FS}}^N \in S_N} \mathcal{G}(U \ket{\psi}_{\text{FS}}^N),
   \label{eq:ent_pow_full}
\end{eqnarray}
where $\ket{\psi}^{N}_\text{FS} = \bigotimes\limits_{i=1}^{N} \ket{\psi_i}$ represents a fully separable state and the maximization is performed over the set of FS states, denoted by $S_N$. Likewise, the ability of a unitary operator to reduce genuine multipartite entanglement defines its \textit{multipartite disentangling power}. This quantity measures how effectively a unitary operator can transform an initially genuine multipartite entangled  state into a FS state. It can also be rewritten as the entangling power of an adjoint unitary operator, given by
\begin{eqnarray}
    \mathcal{D}_N(U) \equiv  \mathcal{E}_N({U}^{\dagger}) = \max_{\ket{\psi}_{\text{FS}}^N \in S_N} \mathcal{G}({U}^{\dagger} \ket{\psi}_{\text{FS}}^N),
   \label{eq:diss_pow_full}
\end{eqnarray}
where the maximization is again performed over the set \( S_N \). Note that instead of FS states, both the definition can be extended by modifying the set for optimization (see Appendix \ref{sec:ggm}) \cite{hazra2023hierarchiesgenuinemultipartiteentangling,mondal2023imperfectentanglingpowerquantum}. Intuitively, it seems that all unitary operators have equal $\mathcal{E}_N(U)$ and $D_N(U)$. However, it was shown \cite{Linden2009_prl} that for bipartite systems, there exists a class of  unitary operators for which $\mathcal{E}_2(U)\neq \mathcal{D}_2(U)$ . Hence a natural question arises whether such a unitary operator and its corresponding adjoint exist in a multipartite regime which can exhibit unequalness. In this work, we construct such classes of unitary operators. 



\subsection{Alike entangling and disentangling power for diagonal unitary operators}
\label{sec:diagonal}

Let us first prove  that entangling and disentangling capabilities of diagonal unitary operators acted on \(N\)- party systems remain same $\forall ~ N$. To exhibit this analytically, we first show this equivalence for a specific class of eight-dimensional diagonal unitary operators and then extend it for  \(2^N\)-dimensional diagonal unitaries. We then argue that such similarity between entangling and disentangling powers holds for arbitrary diagonal unitaries, which is supported by our  numerical simulations (Appendix \ref{sec:diag_app1}).

{\it Three-qubit diagonal unitaries.} To initiate this study, we begin with a single-parameter family of diagonal unitaries, \(U_{d,\phi} = \mathrm{diag}(1,1,1,1,1,1,1,e^{i\phi})\) ($\phi \in \mathbb{R} $), and evaluate their entangling and disentangling powers with respect to fully separable input states. 

\begin{proposition}
\label{prop:diag_fullsep}

The entangling, \(\mathcal{E}_3(U_{d,\phi})\), and disentangling \(\mathcal{D}_3(U_{d,\phi})\), powers  coincide for \(U_{d,\phi}\).
\end{proposition}

\begin{proof}
The action of the three-qubit diagonal unitary operator \(U_{d,\phi}\) on a fully separable input state, \(|\psi\rangle^{3}_{\mathrm{FS}} = \bigotimes\limits_{i=1}^{3} |\psi_i\rangle\), where  \(|\psi_i\rangle = \cos \frac{\theta_i}{2} |0\rangle + e^{i \xi_i} \sin\frac{\theta_i}{2} |1\rangle\) leads to the resulting output state,
\begin{eqnarray}
   |\psi\rangle_{out}^{3} &=& U_{d,\phi} |\psi \rangle_{FS}^{3}=\sum_{i_1,i_2,i_3=0}^1 a_{i_1i_2i_3} e^{i \phi_{i_1i_2i_3}}|i_1i_2i_3\rangle ,\nonumber\\
\label{eq:sep3_out}
\end{eqnarray}
where \(e^{i \phi_{i_1i_2i_3}} =1\) except \(e^{i \phi_{111}} = e^{i \phi}\) and  \(\{a_{000}=\prod_{i=1}^{3}\cos \frac{\theta_{i}}{2}, a_{001}=e^{i\xi_{3}} \prod_{i=1}^{2} \cos \frac{\theta_{i}}{2} \sin \frac{\theta_3}{2}, \ldots, a_{111}=e^{i \phi} e^{i \sum_{i=1}^{3}\xi_{i}} \prod_{i=1}^{3}\sin \frac{\theta_{i}}{2}\}\).
We first notice that eigenvalues of the local density matrices required to compute GGM do not depend on $\xi_i$ of \(\ket{\psi_i}\) and hence, we set  \(\xi_i = 0\). Therefore, the entangling power turns out to be
\begin{eqnarray}
\mathcal{E}_{3}(U_{d,\phi}) = \max_{\ket{\psi}^3_{\text{FS}} \in S_3} \mathcal{G}(U_{d,\phi} \ket{\psi}^3_{\text{FS}})  = \max_{\theta}(\frac12 - \frac{1}{32}\sqrt{A}),\nonumber\\
\end{eqnarray}
where \( A =218 + 16 \cos\theta + 49 \cos2 \theta - 24 \cos 3 \theta - 10 \cos 4 \theta + 8 \cos 5 \theta - \cos 6 \theta - 1024 \cos^4\frac{\theta}{2} \left(-3 + \cos\theta\right) \cos \phi \sin^6\frac{\theta}{2}\), and the second equality is written by noticing that the optimization occurs for $\theta_1=\theta_2=\theta_3=\theta$ (say). Due to the presence of the \(\cos \phi\) term in the expression for \(A\), and since $\cos(-\phi) = \cos \phi$, \(\mathcal{G}(U_{d,\phi}\ket{\psi}_{FS}^3) = \mathcal{G}(U_{d,\phi}^\dagger \ket{\psi}_{FS}^3)\), even before the maximization over the set of states, and hence, $\mathcal{D}_3(U_{d,\phi}) = \mathcal{E}_3(U_{d,\phi})$.
The detailed calculation is in Appendix \ref{sec:pro_1}. Hence, the proof.
\end{proof}


{\it Entanglement - generating and - destroying capabilities of a family of arbitrary diagonal unitaries in arbitrary dimension.} Let us first consider \(U_{d,\phi}=\mathrm{diag}(1,\ldots,e^{i\phi})\) which acts on \(N\)- party FS states and the relation between entangling and disentangling powers remains same as shown in Proposition 2 below in a similar manner as in Proposition 1.
\begin{proposition}
\label{prop:alldiag_fullsep}
The entangling, \(\mathcal{E}_N(U_{d,\phi})\) and the disentangling, \(\mathcal{D}_N(U_{d,\phi})\), powers are same for \(U_{d,\phi}\).
\end{proposition}
\begin{proof}
The output state after \(U_{d,\phi}\) acts on \(|\psi\rangle_{FS}^{N} = \overset{N}{\underset{i=0}{\bigotimes}}\ket{\psi_i}\) takes the form as 
  \(|\psi\rangle_{out}^{N} = U_{d,\phi} |\psi \rangle_{FS}^{N}=\sum_{i_1,i_2,\ldots,i_N=0}^1 a_{i_1i_2\ldots i_N} e^{i \phi_{i_1i_2\ldots i_N}}|i_1i_2\ldots i_N\rangle\)
where \(e^{i \phi_{i_1i_2\ldots i_N}} =1\) except \(e^{i \phi_{11\ldots 1}} = e^{i \phi}\) and the set of coefficients can be written as  \(\{a_{00\ldots 0}=\prod_{i=1}^{N}\cos \frac{\theta_{i}}{2}, a_{00\ldots 1}=e^{i\xi_{3}} \prod_{i=1}^{N-1} \cos \frac{\theta_{i}}{2} \sin \frac{\theta_N}{2}, \ldots, a_{11\ldots 1}=e^{i \phi} e^{i \sum_{i=1}^{N}\xi_{i}} \prod_{i=1}^{N}\sin \frac{\theta_{i}}{2}\}\).
While computing GGM of \(\ket{\psi}_{out}^N\), we again set \(\xi_{i}=0\) as it was found numerically upto \(N=8\) that local phase do not contribute. We find that  \(\mathcal{E}_{N}(U_{d,\phi})=\underset{\ket{\psi}^N_{\text{FS}} \in S_N}{\max} \mathcal{G}(U_{d,\phi} \ket{\psi}^N_{\text{FS}}) =\underset{\theta}{\max}(\frac12 - \frac{1}{2^{2N-1}}\sqrt{A})\) where \(A=\sum_{k=0}^{2N} \alpha_k(N) \cos{k\theta} \nonumber- \beta(N) \cos^4{\frac{\theta}{2}}\cos{\phi} \sin^{2N}{\frac{\theta}{2}}\sum_{k=0}^{N-2}\gamma_k(N) \cos^{k}{\theta}\) which  only depends on \(\cos \phi\) implying \(\mathcal{D}_{N}(U_{d,\phi}) = \mathcal{E}_{N}(U_{d,\phi})\). The detailed calculation is in  Appendix~\ref{sec:pro_2}. Hence, the proof.
\end{proof}

{\it Remark.} Numerical simulations suggest that
for the diagonal unitary operators, \(U_{d,\{\phi_i\}} = \mathrm{diag}(e^{i\phi_1}, e^{i\phi_2}, \ldots, e^{i\phi_8})\), \(U_{d,\{\phi_i\}} = \mathrm{diag}(e^{i\phi_1}, e^{i\phi_2}, \ldots, e^{i\phi_{2^4}})\) and \(U_{d,\{\phi_i\}} = \mathrm{diag}(e^{i\phi_1}, e^{i\phi_2}, \ldots, e^{i\phi_{2^5}})\) acted on three-, four- and five-qubit inputs respectively, the entangling and  disentangling powers remain the same when \(\{\phi_i\}\)s are sampled randomly from normal distributions (see Appendix~\ref{sec:num_same}).

\section{Unequal multipartite entangling and disentangling powers}
\label{sec:inequivalence}

In the previous section, we observed that for diagonal unitary operators, the entangling and disentangling powers are identical. By constructing non-diagonal unitary operators, we now exhibit the inequivalence between multipartite entangling and disentangling powers.

\begin{figure}
\includegraphics [width=\linewidth]{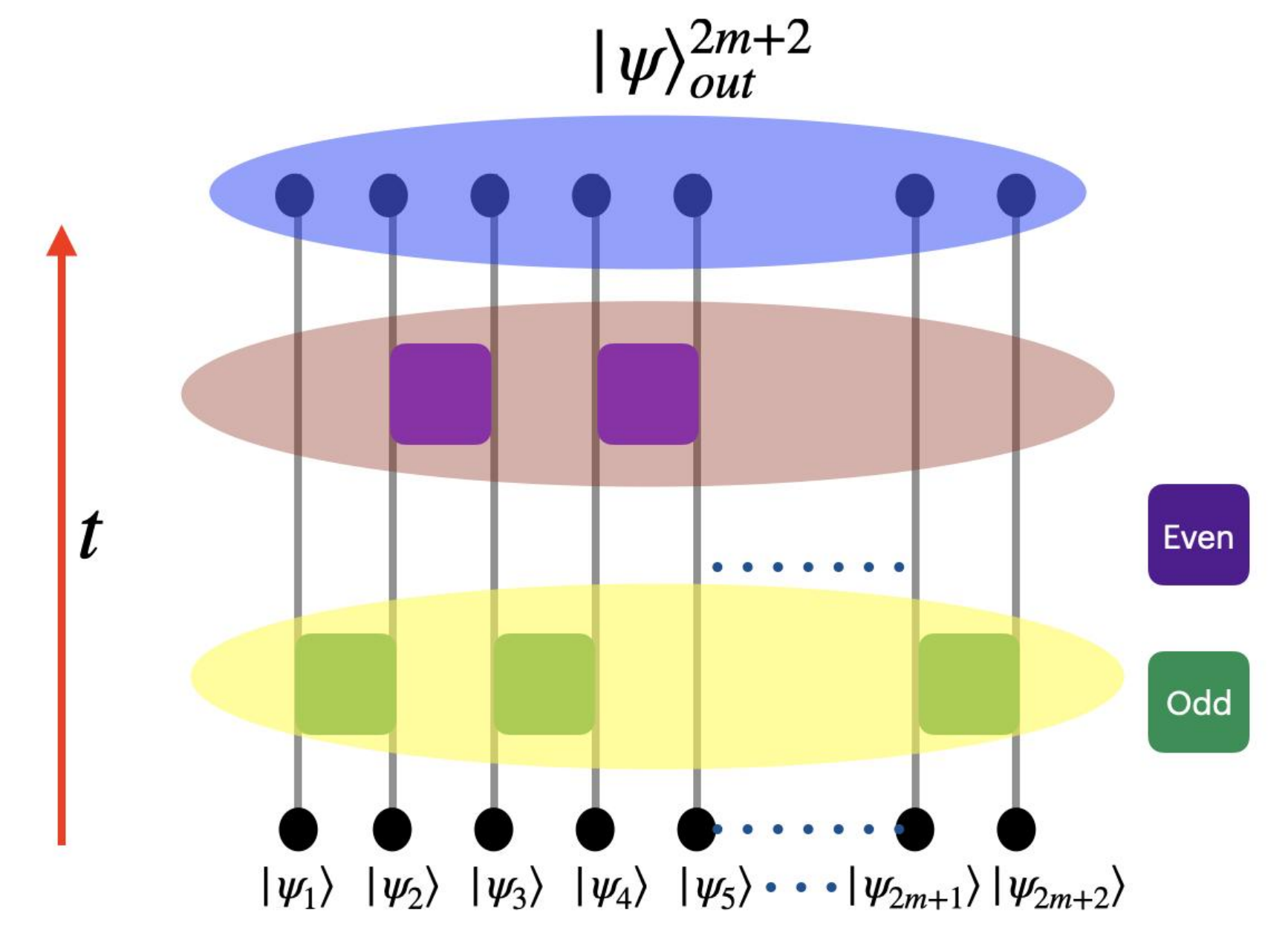}
\caption{{\bf Schematic diagram of the circuit.} 
Illustration of a brickwork quantum circuit with alternating layers of two-qubit unitaries applied in an odd-even pattern. The system begins with a suitable product state, 
\( \ket{\psi}_{in}^{2m+2} = |\psi_1\rangle \otimes |\psi_2\rangle \otimes \cdots \otimes |\psi_{2m+2}\rangle \), 
and evolves through successive applications of two-qubit gates on \text{odd} (green) and \text{even} (purple) sites. After applications of unitaries in these layers at different time steps, the system arrives at a final  state, \( |\psi\rangle_{out}^{2m+2} \) which can, in general, be genuinely multipartite entangled. Similarly, one can apply \(U^{\dagger}\) at each step. Note that we exhibit that there can be a situation where unitaries applied in the even and odd sites can be different and may not commute. 
Time, \(t\), flows upward in the diagram.}
\label{fig:schematic}
\end{figure}

\subsubsection{Construction of non-diagonal unitaries: Even vs odd}
\label{sec:nd_realization}

It is known that in the case of unitary operators acting on a qubit-qutrit input product state, the entangling power can reach its maximum, while the disentangling power remains strictly lower \cite{Linden2009_prl} than the maximum, showing their disparity from the perspective of entanglement generation. As the number of parties increases, determining whether the entangling and disentangling powers can differ becomes increasingly nontrivial due to the increase of dimension. Let us first show that they are unidentical when the number of sites on which they act is even.\\
\\
\textbf{Construction \(1.\)} {\it Let us design a family of non-diagonal unitary operators, \(U_{ND}(\lambda)\) for which the entangling power, \(\mathcal{E}^{\mathrm{even}}_{N}(U_{ND}(\lambda))\) and disentangling power, \(\mathcal{D}^{\mathrm{even}}_{N}(U_{ND}(\lambda))\) are not equal where \(U_{ND}(\lambda)\) acts on an optimal FS states having even number of parties.}

To demonstrate this, we consider a two-qubit unitary operator, which acts on two neighboring parties (see Fig.~\ref{fig:schematic}) as
\begin{eqnarray}
  \nonumber U_{i,i+1}(\lambda)&=&\ket{00}\bra{00}+  \ket{11}\bra{11}+\cos{\lambda}\big(\ket{01}\bra{01} + \ket{10}\bra{10}\big) \nonumber\\&+&\sin{\lambda}\big(\ket{01}\bra{10}-\ket{10}\bra{01} \big),\forall i \in\{1,\ldots, N-1\},\nonumber\\
     \label{eq:so_3}
\end{eqnarray}
where \(0\leq\lambda\leq 2\pi\). Using this, a global unitary operator can be constructed as
\begin{eqnarray}
    U_{ND}(\lambda)=\bigg(\prod_{i\in odd} U_{i,i+1}(\lambda)\bigg)\bigg(\prod_{i\in even} U_{i,i+1}(\lambda)\bigg),
    \label{eq:unitary_even}
\end{eqnarray}
which after the action on the input states with even number of parties, \(N=2m+2\) (\(m\) is a non-negative integer), leads to the output state,
\(|\psi\rangle_{out}^{\mathrm{even}}=U_{ND}(\lambda)|\psi\rangle_{in}\), where \(|\psi\rangle_{in}=\bigotimes\limits_{i=1}^{2m+2}(\cos{\theta_i}\ket{0} + e^{i\xi_i}\sin{\theta_i}\ket{1})\) and similarly after \(U_{ND}^{\dagger}(\lambda)\). Interestingly, we observe that the optimal input states in the cases of \(U_{ND}(\lambda)\) and \(U_{ND}^{\dagger}(\lambda)\) have some specific structures. Specifically, in the case of \(U_{ND}(\lambda)\), \(|\psi\rangle_{in}^{opt}=|\psi_1\rangle \overset{2m+1}{\underset{i=2}{\bigotimes}}|\psi_i\rangle \otimes |\psi_{2m+2}\rangle\) where \(|\psi_1\rangle=|\psi_{2m+2}\rangle=\cos\frac{\theta_1}{2} |0\rangle + \sin \frac{\theta_1}{2} |1\rangle\) and \(|\psi_i\rangle=\cos \frac{\theta_i}{2} |0\rangle + \sin \frac{\theta_i}{2} |1\rangle\)  while for \(U_{ND}^{\dagger}(\lambda)\), the optimal state is \(|\psi\rangle_{in}^{opt}=\underset{i\in odd}{\bigotimes}|\psi_{i}\rangle \underset{j\in even}{\bigotimes}|\psi_{j}\rangle\) where \(i\) and \(j\) denote odd and even integer sites, respectively and \(|\psi_i\rangle=\cos\frac{\theta}{2} |0\rangle + \sin \frac{\theta}{2} |1\rangle\) and similarly \(|\psi_j\rangle\) with \({\theta}^{'}\). With these symmetries in hand, optimization over the \(2m+2\) parameters in the case of \(U_{ND}^{\dagger}(\lambda)\) reduces to two parameters, although such reduction is not possible for \(U_{ND}(\lambda)\) which still involves ($2m+1$)-parameter optimization. After performing this maximization numerically, we report that \(\mathcal{E}_{N}^{\mathrm{even}}(U_{ND}(\lambda))\neq\mathcal{D}_{N}^{\mathrm{even}}(U_{ND}(\lambda))\) for a finite range of \(\lambda\), especially when \(\lambda \in (\frac{\pi}{4}, \frac{\pi}{2})\) and \(\lambda \in (\frac{\pi}{2}, \frac{3\pi}{4})\) (for illustration, see Fig.~\ref{fig:diff_ggm_all} for four- and six- parties). The similar behavior with increasing number of sites strongly suggests  that this remains true for an arbitrary even \(N\) that  we also check numerically upto \(N=10\). 
Interestingly, we find that when the number of parties on which the unitary operator acts is odd, \(U_{ND}(\lambda)\) and \(U_{ND}^{\dagger}(\lambda)\) in Eq.~(\ref{eq:unitary_even}) possess identical entangling and disentangling powers. Note further that for inputs with even and odd number of parties, \(U_{ND}(\lambda)\) in Eq. (\ref{eq:unitary_even}) contains odd and even number of unitaries respectively, thereby illustrating a distinct differences between their capabilities. To overcome this hurdle and to make the entire exercise independent of the dimension, we identify another class of non-diagonal unitary operators, which serve  the purpose when they act on odd number of parties.\\

\textbf{Construction \(2.\)} Let us illustrate the structure of the unitary operator, $U_{ND}(\lambda)$ which operates on odd numbers of qubits to provide unidentical entangling and disentangling powers. First, we define a two-qubit unitary denoted as $U_w$, acting  on \((2m,2m+1)\)-qubit pair by using two nonorthogonal states, \(\ket{\beta} =\frac{\sqrt{3} }{2}\ket{1}-\frac{1}{2}\ket{0}, \ket{\gamma}  =-\frac{\sqrt{3}}{2}\ket{1}-\frac{1}{2}\ket{0}\), resulting in a set of orthonormal two-qubit states,
 \(\ket{ w_{01}} =\frac{\omega ^2 \ket{\gamma}\ket{1} +\omega  \ket{\beta} \ket{0}]}{\sqrt{2}}, \ket{w_{00}}=\frac{\ket{\gamma}\ket{1}+\ket{\beta }\ket{0}}{\sqrt{2}}, 
\ket{w_{10}}=\frac{\ket{\gamma_t}\ket{1}+\ket{\beta_t}\ket{0}}{\sqrt{2}}\) and \(\ket{w_{11}}=\frac{\omega ^2 \ket{\gamma_t}\ket{1}+\omega  \ket{\beta_t}\ket{0}}{\sqrt{2}}\) with \(\omega=e^{i\pi}\) and $\ket{\beta_t}$ and $\ket{\gamma_t}$ being the orthogonal states of $\ket{\beta}$ and $\ket{\gamma}$ respectively.
Based on these basis states, the two-qubit unitary operator $U_w$ takes the form \cite{Linden2009_prl} as 
\begin{eqnarray}
\nonumber U_w &=& \ket{w_{01}}\bra{01}+\ket{w_{10}}\bra{10}+\ket{w_{11}}\bra{11}-i \ket{w_{00}}\bra{00},
\label{eq:winter}
\end{eqnarray}
which leads to a family of a unitary operators in the \(2^{2m+1}\)-dimensional space as
\begin{eqnarray}
    U_{ND}(\lambda)=\bigg(\underset{i\leq 2m}{\underset{i\in odd}{\prod}} U_{i,i+1}(\lambda)\bigg)\bigg(\underset{i<2m}{\underset{i\in even}{\prod}} U_{i,i+1}(\lambda)\bigg)U_w.
    \label{eq:unitary_odd}
\end{eqnarray}
Here, we assume that \(U_w\) acts on the last two sites, \(2m\) and \(2m+1\) which can be any two neighboring qubits. Note that instead of \(U_w\), one can generate any four-dimensional random unitary and construct \(U_{ND}(\lambda)\) in Eq.~(\ref{eq:unitary_odd}). In all these situations, our numerical maximization reveals that their exists a range of \(\lambda\) for which \(\mathcal{E}_{2m+1}^{\mathrm{odd}}\big(U_{ND}(\lambda)\big)\neq\mathcal{D}_{2m+1}^{\mathrm{odd}}\big(U_{ND}(\lambda)\big)\). In the case of three-qubit FS inputs, when\(\frac{\pi}{4}\lesssim \lambda <\pi\) and \(\pi< \lambda \lesssim\frac{3\pi}{4}\), the above disparity can be found although the range changes slightly with the number of parties involved (see Fig.~\ref{fig:diff_ggm_all} \((c)\) and \((d)\)). Note that instead of this systematic construction, if one generates random unitary operators in respective dimensions, the entangling and disentangling powers also become unequal (see Appendix~\ref{sec:inequ_app_c}). Such a precise formulation of unitaries have some benefits from the perspective of the implementation in different physical platforms as will be demonstrated next.

\begin{figure}
\includegraphics [width=\linewidth]{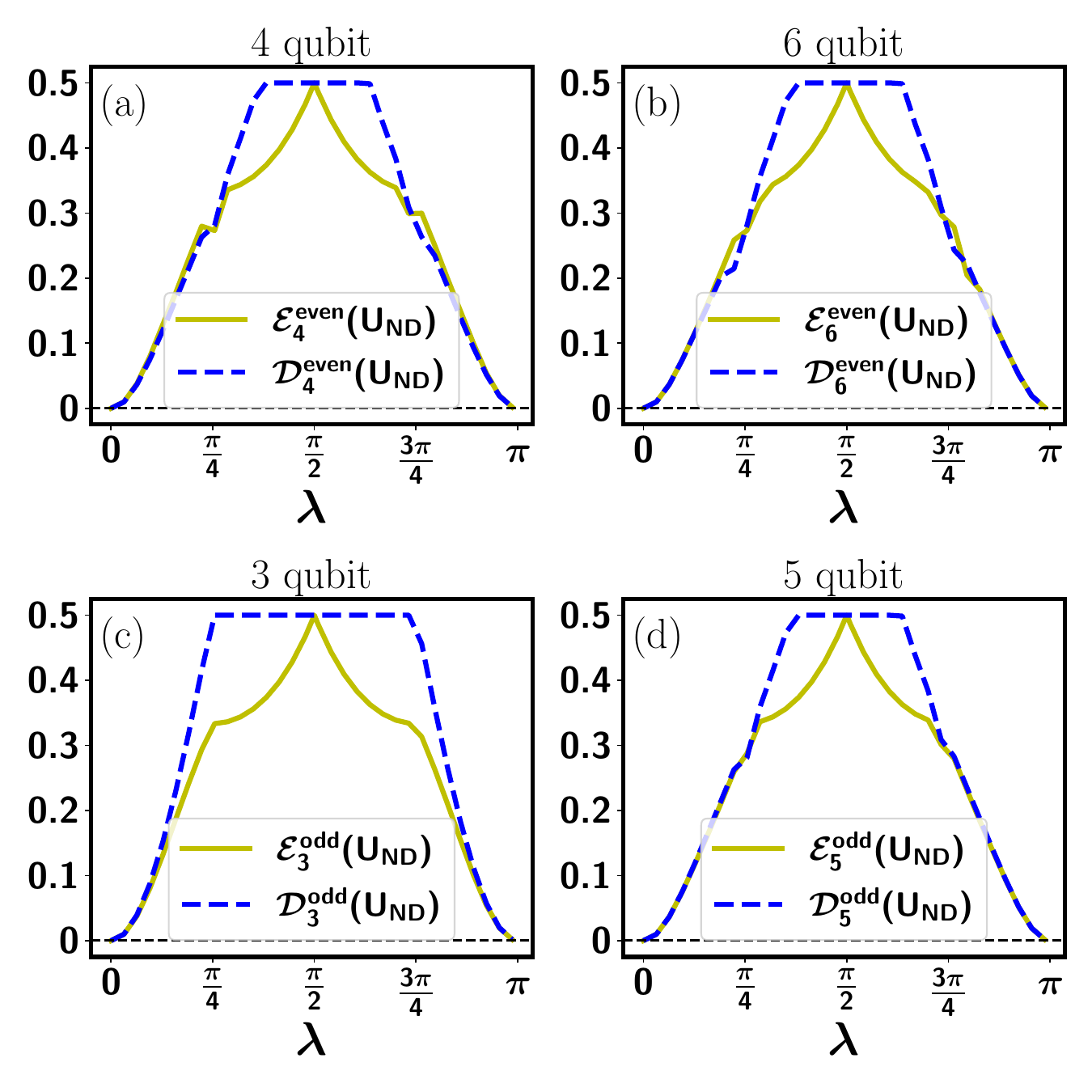}
\caption{\textbf{Inequivalence between multipartite entangling and disentangling powers for non-diagonal unitaries.}  Plot of $\mathcal{E}_{N}^{\mathrm{even(odd)}}(U_{ND}(\lambda))$ (solid line), and $\mathcal{D}_{N}^{\mathrm{even(odd)}}(U_{ND}(\lambda))$ (dashed line) (ordinate) against $\lambda$ (abscissa) for non-diagonal unitary operators $U_{ND}(\lambda)$. (a) and (b) correspond to the output state with even number of qubits, $N = 4$ and $N = 6$ respectively after the action of $U_{ND}(\lambda)$ in Eq. (\ref{eq:unitary_even}). (c) and (d) represent the same for odd number qubits, $N = 3$ and $N = 5$ respectively with $U_{ND}(\lambda)$ in Eq. (\ref{eq:unitary_odd}) . Both axes are dimensionless.}
  
\label{fig:diff_ggm_all}
\end{figure}

\subsubsection{Simulating unitaries through interacting Hamiltonians}
\label{sec:unitaryHamil}

Let us illustrate how the unitary operator, described in Eqs.~(\ref{eq:unitary_even}) and (\ref{eq:unitary_odd}) for even and odd sites, respectively, can be simulated on the currently available experimental platforms, like cold atoms in optical lattices and trapped ions \cite{Cirac1995_prl, Molmer1999, Porras2004, Bruzewicz2019, Michael2024, Jaksch1999, Friesen2007OnewayQC, Dai2016, Schfer2020, Zhang2023}.  
{\it Dynamics with even number of sites. } We first present our results for an even number of sites, followed by the system with odd sites. Starting from the initial optimal product state described before for $U_{ND}(\lambda)$ and $U_{ND}^\dagger(\lambda)$, \footnote{ To compute \(\mathcal{E}_{N}^{\mathrm{even}}(U_{DM})\) and \(\mathcal{D}_{N}^{\mathrm{even}}(U_{DM})\), we again optimize over the set of fully separable initial states, and in this case, they turn out to be same as obtained for \(U_{ND}(\lambda)\) and \(U_{ND}^{\dagger}(\lambda)\) with \(N=4, 6\).}  a system evolves under a unitary operator, described by \cite{Dzyaloshinsky1958, Moriya1960, Moriya1960_prl, Jafari2008} (see Fig.~\ref{fig:schematic})
\begin{eqnarray}
    U_{\mathrm{DM}} &=& e^{-iH^{\mathrm{odd}}_{\mathrm{DM}}t}e^{-iH^{\mathrm{even}}_{\mathrm{DM}}t},
\end{eqnarray}
where the Hamiltonians involve Dzyaloshinskii–Moriya (DM) interaction \cite{Dzyaloshinsky1958, Moriya1960}, represented as 
\begin{eqnarray}
H^{\mathrm{odd}}_{\mathrm{DM}} &=& \frac{1}{2}\sum_{i\in\mathrm{odd}} (\sigma^y_i \sigma^x_{i+1} - \sigma^x_i\sigma^y_{i+1}),
\end{eqnarray}
with $\sigma^\alpha (\alpha = x,y,z)$ being the Pauli matrices. Similarly, we can construct \(H^{\mathrm{even}}_{\mathrm{DM}}\), when \(i \in \mathrm{odd}\).
 Note that \(H^{\mathrm{even}}_{\mathrm{DM}}\) and \(H^{\mathrm{odd}}_{\mathrm{DM}}\) do not commute. Interestingly, to present \(\mathcal{E}_{N}^{\mathrm{even}}(U_{\mathrm{DM}})\neq \mathcal{D}_{N}^{\mathrm{even}}(U_{DM})\), we introduce a quantity, \(\Delta^{\mathrm{DM}}_{N}=|\mathcal{E}(U_{\mathrm{DM}})-\mathcal{D}(U_{\mathrm{DM}})|\) quantifying their unequalness and find its  nonvanishing and nonmonotonic behaviors with respect to \(0\leq t \leq \pi\) except \(t=\frac{n\pi}{2}\) (\(n=0,1,2\)) (see Fig.~\ref{fig:diff_even_odd_hm}).
 
\begin{figure}
\includegraphics [width=\linewidth]{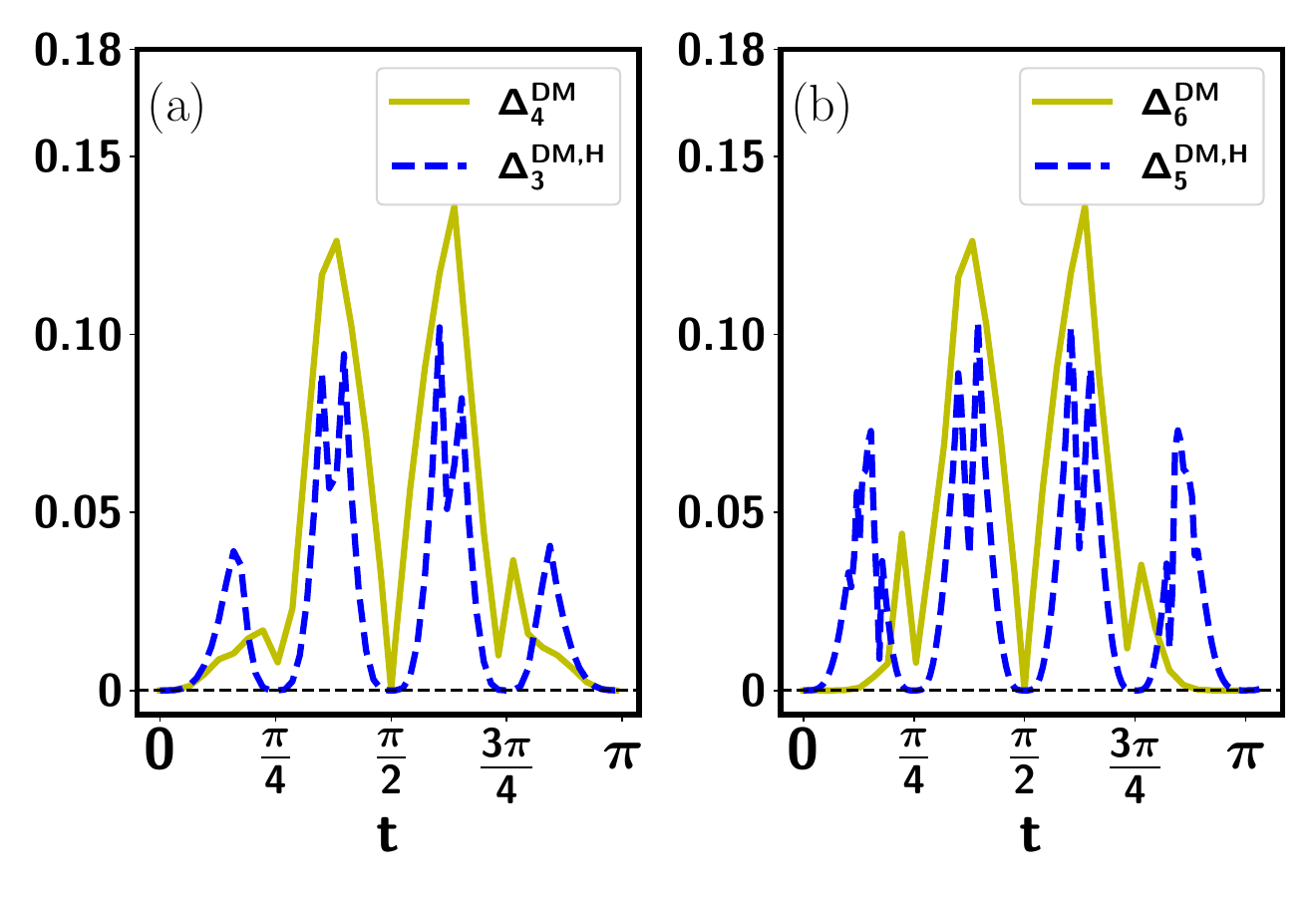}
\caption{ \textbf{Absolute difference between entangling and disentangling powers under engineered Hamiltonians.}  $\Delta^{\mathrm{DM}}_{N}$ and $\Delta^{\mathrm{DM,H}}_{N}$ (ordinate) with respect to time $t$ (abscissa) for different system sizes. (a) shows results for $N = 3$ (dashed line) and $N = 4$ (solid line), while (b) corresponds to $N = 5$ (dashed) and $N = 6$ (solid). For even sites, $N = 4,\, \text{and}\,\,6$, the dynamics is governed by the  Hamiltonians involving NN Dzyaloshinskii–Moriya (DM) interactions while for odd $N = 3,\, \text{and}\, 5$ sites, the evolution is driven by Hamiltonians with DM and Heisenberg interactions. In all the cases, the initial states are  fully separable which are chosen appropriately. Clearly, multipartite entangling and disentangling powers differ.  Both axes are dimensionless.
}
\label{fig:diff_even_odd_hm} 
\end{figure}

{\it Governing Hamiltonian involving odd sites.} As it was shown that the unitaries which act on odd number of qubits showing dissimilarities are much more involved than the ones for even number of qubits. This can also be visible in the interacting Hamiltonian.
The $(i,i+1)$- pair of the system (with $i$ being the even integer) evolves accoding to the Heisenberg Hamiltonian given by $H_H^{\mathrm{even}}=\underset{i\in even}{\sum}(\sigma^{x}_{i}\sigma_{i+1}^{x}+\sigma^{y}_{i}\sigma_{i+1}^{y}+\sigma^{z}_{i}\sigma_{i+1}^{z})$, followed by the evolution via DM interaction in the $(i+1,i+2)$-bond, i.e., the evolution operator, in this case, reads as 
\begin{eqnarray}
{U_{\mathrm{DM,H}}}&=&  e^{-iH_{\mathrm{DM}}^{\mathrm{odd}}t}e^{-iH_{H}^{\mathrm{even}}t},
\end{eqnarray}
where $H_{\mathrm{DM}}^{\mathrm{odd}}$ and $H_{H}^{\mathrm{even}}$ also do not commute. To find $\Delta_{N}^{\mathrm{DM,H}}=|\mathcal{E}(U_{\mathrm{DM,H}}) - \mathcal{D}(U_{\mathrm{DM,H}})|>0$ for certain values of the evolution time, \(0\leq t\leq \pi\) except \(t=\frac{n\pi}{4}\) for \(n=0,1,2,3,4\) (see Fig.~\ref{fig:diff_even_odd_hm}), we here notice that \(U_{DM,H}\) and $U_{DM,H}^\dagger$ require different initial FS state than the one obtained in Construction \(2\).\\
\\
{\it \textbf {Randomized circuit design for illustration.}} The entire analysis reveals that when two noncommuting unitary operators are applied successively, disparity between powers emerges. Let us describe a setup of a quantum circuit comprising two-qubit gates having multiple layers (see Fig.~\ref{fig:schematic}).
(i) The first layer consists of two-qubit gates acting on all odd-numbered pairs of qubits, e.g., $(1,2)$-, $(3,4)$-pair etc., forming a single global unitary operator, generated by an interacting Hamiltonian, $H_{\mathrm{odd}} = \underset{i\in odd}{\sum}H^a_{i,i+1}$, where \(H^a_{i,i+1}\) is either the Hamiltonian described before or a random Hermitian operator. (ii) The two-qubit gates in the second layer act on even-numbered pairs, i.e.,  on $(2,3)$-, $(4,5)$- pairs etc., forming another global unitary, described by $H_{\mathrm{even}}=\underset{i\in even}{\sum}H^b_{i,i+1}$ where \(H_{i,i+1}^b\) can be taken in a similar manner as  in \(H_{i,i+1}^a\).\\  
 For systems with an even number of qubits, using the \emph{same} two-qubit Haar-random unitary  or Hamiltonian with DM interaction applied in both  layers already results in two noncommuting global unitaries, leading to a clear asymmetry between their entangling and disentangling powers. However, in this case,  these two powers of the resulting global unitaries remain distinct even when \(H_{i,i+1}^{a}\) and \(H_{i,i+1}^{b}\) are the same or commuting. Surprisingly, when the total number of sites is odd, we need to employ two different sets of Haar-random two-qubit gates for the two layers or unequal (or non-commuting) \(H^a_{i,i+1}\) and \(H_{i,i+1}^{b}\)  so that a clear difference between multipartite entangling and disentangling powers emerge. It highlights that for odd system-size, more stringent criteria are required to observe inequivalence between these two powers.

\section{Conclusion}
\label{sec:conclusio}

In quantum information processing, entanglement is a central resource responsible for providing quantum advantages in protocols, such as dense coding, teleportation, quantum cryptography and one-way quantum computation. Yet, the ability to disentangle - effectively reducing or erasing entanglement - is equally crucial in tasks like decoherence control, environment decoupling, and secure communication. In multipartite systems, where entanglement structures are more complex and less understood, a fundamental question arises: Does a given unitary operator have greater power to entangle or to disentangle? Our investigation addresses this question by revealing that, despite the mathematical reversibility of unitary operations, their physical action on multipartite entanglement can exhibit a pronounced asymmetry. 
This feature also highlights the need to account for directional entanglement dynamics in the design and implementation of quantum processes and devices.

We defined the multipartite entangling and disentangling powers of a unitary operator based on the genuine multipartite entanglement measure. We began by showing analytically that  the entangling and disentangling powers for a specific class of diagonal unitaries — those with just the last diagonal element containing a phase —  are equal. Numerical studies confirmed that this equivalence also holds for general diagonal unitaries in three-, four-, and five- qubit systems. However, when we moved beyond diagonal unitary operators to a more complex, and specially constructed non-diagonal unitaries, we observed a fundamental asymmetry: the ability of a unitary operator to generate entanglement can differ significantly from its ability to remove it. This inequivalence was evident when optimization was performed over the set of fully separable input states. In order to validate this observation in physical systems, we demonstrated that such asymmetry can arise from unitary evolution governed by  Hamiltonians, that involve Dzyaloshinskii–Moriya (DM) and Heisenberg \(XXX\) interactions. Importantly, we observed a qualitatively different behavior depending on whether the number of qubits on which unitary operators act is even or odd. For systems with an even number of qubits, applying two successive non-commuting layers of two-qubit Haar-random unitaries - acting respectively on even- and odd- numbered pairs - was sufficient to reveal a clear disparity between entangling and disentangling power. However, in the case of odd-numbered systems, this difference only emerged when the two layers were constructed from \textit{distinct} sets of Haar-random unitaries of two qubits, indicating a more subtle structure in the entanglement dynamics.
This contrast emphasizes the relevance of the system size and pairing structure in determining how entanglement is formed or suppressed during quantum evolution. 
 This insight could be valuable in the development of quantum technologies that exploit or suppress entanglement.

\section*{Acknowledgements}
We acknowledge the use of \href{https://github.com/titaschanda/QIClib}{QIClib} -- a modern C++ library for general purpose quantum information processing and quantum computing (\url{https://titaschanda.github.io/QIClib}) and cluster computing facility at Harish-Chandra Research Institute.

\appendix

\section{Characterizing genuine multipartite entanglement }
\label{sec:gme}

In an \( N \)- partite quantum system, a pure state \( \ket{\psi} \) is said to be \textit{\( k \)-separable} (for \( k = 2, \dots, N \)) if it can be expressed as a product of \( k \) pure states, each defined on mutually disjoint subsets of the full system. Formally,
 \( \ket{\psi} = \bigotimes_{m=1}^{k} \ket{\psi_m}\),
where each \( \ket{\psi_m} \) belongs to one of the \( k \) partitions. Such a state is denoted by \( \ket{\psi}_{k\text{-sep}} \).

\begin{itemize}
    \item When \( k = N \), the state is referred to as \textit{fully separable (FS) state}, and, mathematically, reads as $\ket\psi =\overset{N}{\underset{i=0}{\bigotimes}}\ket{\psi_i}$.
    \item A $N$- party state is  \textit{genuinely multipartite entangled (GME)} if it cannot be written as a product across any bipartition, i.e., it is not separable with respect to any division containing two or more parts of the system.
\end{itemize}

\subsection{Quantifying genuine multipartite entanglement: Geometric measures}
\label{sec:ggm}

We quantify genuine multipartite entanglement content of the output states \cite{aditi2010} produced after the action of a unitary operator from a geometrical perspective. Specifically, we compute generalized geometric measure (GGM) which is defined as \(\mathcal{G} (|\psi\rangle) = 1 - \max\limits_{|\phi\rangle \in S_{nG}} |\langle \phi | \psi \rangle|^2 \), where \(S_{nG}\) denotes the set of all non-genuinely multipartite entangled states. We choose this measure as it can easily be computed for pure states in terms of Schmidt coefficients, as $\mathcal{G}(\ket\psi)= 1-\max \{ \{e_{i_1}^m\},\{e_{i_1i_2}^m\},\ldots,\{e_{i_1i_2 \ldots i_{N/2}}^m\}\}$ where maximization is performed over all sets,
\(\{\{e_{i_1i_2}^m\},\ldots,\{e_{i_1i_2 \ldots i_{N/2}}^m\}\}\) containing maximum eigenvalues of \(l\)- site reduced density matrices (\(l=1,2,\ldots N/2\)) \cite{aditi2010, Wei2003_pra, Blasone2008_pra}.

One can define the entangling power of a unitary operator \( U \) with respect to \( k \)-separable states as \(\mathcal{E}_k(U) = \underset{\ket{\psi}_{k\text{-sep}}^N \in S_k}{\max} \mathcal{G}\left(U \ket{\psi}_{k\text{-sep}}^N\right),\)
where \(\ket{\psi}_{k\text{-sep}}^N = \bigotimes_{i=1}^k \ket{\psi_i}\) denotes a \( k \)-separable pure state from the set \( S_k \). By definition, \(\mathcal{G}(\ket{\psi}) = 0\) for all \(\ket{\psi} \in S_k\). Similarly, the disentangling power is defined as 
\(\mathcal{D}_k(U) = \underset{\ket{\psi}_{k\text{-sep}}^N \in S_k}{\max} \mathcal{G}\left(U^\dagger \ket{\psi}_{k\text{-sep}}^N\right).\)

\section{Entangling and disentangling power of diagonal unitaries}
\label{sec:diag_app1}
\subsection{Proof of Proposition 1.}
\label{sec:pro_1}
The proof proceeds by considering the output state given in Eq. (\ref{eq:sep3_out}).
To calculate the entangling power \(\mathcal{E}_3(U_{d,\phi})=\underset{\ket{\psi}^3_{\text{FS}} \in S_3}{\max} \mathcal{G}(U_{d,\phi} \ket{\psi}^3_{\text{FS}})\), we have to calculate   the reduced density  matrix \(\rho_i\) for each party of the given state in Eq. (\ref{eq:sep3_out}) and  $\rho_i$, can take the form as 
\begin{equation}
    \rho_i = \begin{pmatrix}
    a_i & b_i\\
    b^*_i & c_i
    \end{pmatrix},
\end{equation}
where $a_i,b_i$ and $c_i$ are as follows:
\begin{eqnarray}
 \nonumber   a_i &=&{\cos^2\frac{\theta_i}{2}} ,\\
  \nonumber  b_i &= &\cos\frac{\theta_i}{2} \sin\frac{\theta_i}{2} \left( \sum_{\substack{j=1 \\ j \neq i}}^{3-\min(i,2)}
 \cos^2\frac{\theta_j}{2} + r_1\sin^2\frac{\theta_j}{2}  \right),
\\
    c_i&=& {\sin^2\frac{\theta_i}{2}},
\end{eqnarray}
where \(r_1=\left( \overset{3}{\underset{\substack{k=1 \\ k \neq j,i}}{\sum}} \cos^2\frac{\theta_k}{2}+ e^{-i\phi} \sin^2\frac{\theta_k}{2} \right)\).
The eigenvalues $\lambda_\pm$ of $\rho_i$ are given by 
\begin{equation}
    \lambda_\pm =\frac{1}{2}\pm\frac{\sqrt{1-4(a_ic_i-|b_i|^2)}}{2}.
    \label{eq:eigval}
\end{equation}
To maximize \(\mathcal{E}_{3}(U_{d,\phi})\), the term  $(a_ic_i-|b_i|^2)$ must be maximized, i.e.,
\begin{eqnarray}
  \nonumber  \frac{d(a_ic_i-|b_i|^2)}{d\theta_j} &=&0, \forall i,j\in\{1,2,3\}\\  \frac{d^2(a_ic_i-|b_i|^2)}{d\theta_j^2} &<&0.
  \label{eq:der_3qubit}
\end{eqnarray}
From Eq.~(\ref{eq:der_3qubit}), one can arrive at  conditions,
\begin{eqnarray}
\nonumber 1+\cos{\theta_{i+1}}+\cos{\theta_{i+2}}-\cos{\theta_{i+1}}\cos{\theta_{i+2}} = 0\\ \forall i \in\{1,2,3\},
     \label{eq:con_der1}
\end{eqnarray}
and
\begin{eqnarray}
  \nonumber  -\frac{1}{2} \sin^2\theta_i \sin^2\theta_{i+1} \sin^4\frac{\theta_{i+2}}{2} \sin^2\frac{\phi}{2}<0,
 \\\nonumber \forall i\in \{1,2,3\}.
\end{eqnarray}
Note that from Eq. (\ref{eq:con_der1}),
we reach at the condition  \(\cos\theta_1=\cos\theta_2=\cos\theta_3\), i.e., \(\theta_1=\theta_2=\theta_3\) and the eigenvalues takes the form as 
\begin{eqnarray}
\lambda_\pm&=&\frac{1}{2}\pm \frac{1}{32}\sqrt{A}.
    \label{eq:eigval_3sep}
\end{eqnarray}
Here \(A\) can be expressed as
\begin{widetext}
\begin{eqnarray}
 \nonumber A &=&218 + 16 \cos\theta + 49 \cos2 \theta - 24 \cos3 \theta - 10 \cos 4 \theta + 8 \cos 5 \theta - \cos 6 \theta   - 1024 \cos^4\frac{\theta}{2}\left(-3 + \cos\theta\right) \cos\phi \sin^6\frac{\theta}{2}.
   \label{eq:ggm_A}
\end{eqnarray} 
\end{widetext}
Thus the entangling power, \(\mathcal{E}_3(U_{d,\phi})=\underset{\theta}{\max}(\frac{1}{2}-\frac{1}{32}\sqrt{A})\).
On the other hand, the output state after \({U^{\dagger}_{d,\phi}}\) acts on \(|\psi\rangle^3_{FS}\) takes the form as
\begin{eqnarray}
   \nonumber |\psi\rangle_{out}^{3} &=& U^{\dagger}_{d,\phi} |\psi \rangle_{FS}^{3}=\sum_{i_1,i_2,i_3=0}^1 a_{i_1i_2i_3} e^{-i \phi_{i_1i_2i_3}}|i_1i_2i_3\rangle ,
\label{eq:sep_3_output}
\end{eqnarray}
where \(e^{-i \phi_{i_1i_2i_3}} =1\) except \(e^{-i \phi_{111}} \equiv e^{-i \phi}\) and  \(a_{000}=\prod_{i=1}^{3}\cos \frac{\theta_{i}}{2}, a_{001}=e^{i\xi_{3}} \prod_{i=1}^{2} \cos \frac{\theta_{i}}{2} \sin \frac{\theta_3}{2}\) and \(a_{111}=e^{-i \phi} e^{i \sum_{i=1}^{3}\xi_{i}} \prod_{i=1}^{3}\sin \frac{\theta_{i}}{2}\) and so on.
Interestingly, the eigenvalues of the reduced density matrix of the output state, \(|\psi\rangle_{out}^{3}\), are the same as in Eq. (\ref{eq:eigval}). Therefore, we can push similar conditions as in Eq. (\ref{eq:der_3qubit}) and can get the same result as in Eq. (\ref{eq:con_der1}). Therefore, the condition  arises as \(\cos \theta_1 = \cos \theta_2 = \cos \theta_3\), i.e., \(\theta_1=\theta_2=\theta_3\). Thus, the disentangling power, \(\mathcal{D}_3(U_{d,\phi}) =\underset{\theta}{\max}(\frac{1}{2}-\frac{1}{32}\sqrt{A})\). This completes the proof of Proposition \ref{prop:diag_fullsep}.

\subsection{Proof of Proposition 2.}
\label{sec:pro_2}

The proof proceeds by considering the output state, given by \(|\psi\rangle_{out}^{N} = U_{d,\phi} |\psi \rangle_{FS}^{N}=\sum_{i_1,i_2,\ldots,i_N=0}^1 a_{i_1i_2\ldots i_N} e^{i \phi_{i_1i_2\ldots i_N}}|i_1i_2\ldots i_N\rangle\). 
To calculate the entangling power \(\mathcal{E}_N(U_{d,\phi})=\underset{\ket{\psi}^N_{\text{FS}} \in S_N}{\max} \mathcal{G}(|\psi\rangle_{out}^{N})\), we have to calculate   the reduced density  matrix \(\rho_i\) for each party of the given state \(|\psi\rangle_{out}^{N}\) and  $\rho_i$, can take the form as 
\begin{equation}
    \rho_i = \begin{pmatrix}
    a_i & b_i\\
    b^*_i & c_i
    \end{pmatrix},
\end{equation}
The eigenvalues \(\lambda_\pm\) of the reduced matrix \(\rho_i\) is exactly same as Eq. (\ref{eq:eigval}), i.e, \(\lambda_{\pm}=\frac{1}{2}\pm\frac{\sqrt{1-4(a_ic_i-|b_i|^2)}}{2}\) where \(i\in \{1,2,...N\}\) and \(a_i,b_i\) and \(c_i\) take the following forms:
\begin{widetext}
\begin{eqnarray}
 \nonumber   a_i &=&{\cos^2\frac{\theta_i}{2}}, \\
  \nonumber  b_i &= &
 \begin{cases}
\cos\frac{\theta_i}{2} \sin\frac{\theta_i}{2} \left(
 \cos^2\frac{\theta_1}{2} + \sin^2\frac{\theta_1}{2}(\cos^2\frac{\theta_2}{2}+ \sin^2\frac{\theta_2}{2}(\cos^2\frac{\theta_3}{2}+\ldots \sin^2\frac{\theta_{i-1}}{2}(\cos^2\frac{\theta_{i+1}}{2}+\ldots \sin^2\frac{\theta_{N-1}}{2}(\cos^2\frac{\theta_{N}}{2}+e^{i\phi}\sin^2\frac{\theta_{N}}{2})\ldots \right),\\ \quad \,\,\,\,\,\,\,\,\,\,\,\,\,\,\text{for } i \neq N \\
\cos\frac{\theta_N}{2} \sin\frac{\theta_N}{2} \left(
 \cos^2\frac{\theta_1}{2} + \sin^2\frac{\theta_1}{2}(\cos^2\frac{\theta_2}{2}+ \sin^2\frac{\theta_2}{2}(\cos^2\frac{\theta_3}{2}+\ldots + \sin^2\frac{\theta_{N-2}}{2}(\cos^2\frac{\theta_{N-1}}{2}+e^{i\phi}\sin^2\frac{\theta_{N-1}}{2})\ldots \right),~~~~ \text{for } i = N
\end{cases}
\\
 c_i&=& {\sin^2\frac{\theta_i}{2}}.
\end{eqnarray}
\end{widetext}
\textbf{Note:} We numerically find out that \(\theta_1=\theta_2=\theta_3=...=\theta_N = \theta, \, \text{say}\), and the eigenvalues read as
\begin{eqnarray}
\lambda_{\pm}&=&\frac12\pm\frac{1}{{2}^{2N-1}}\sqrt{A},
    \label{eq:N_ggm}
\end{eqnarray}
\begin{eqnarray}
   \nonumber \text{where}\, A&=&\sum_{k=0}^{2N} \alpha_k(N) \cos{k\theta} \\\nonumber&&- \beta(N) \cos^4{\frac{\theta}{2}}\cos{\phi} \sin^{2N}{\frac{\theta}{2}}\sum_{k=0}^{N-2}\beta_k(N) \cos^{k}{\theta}.
\end{eqnarray}
Thus, the entangling power, \(\mathcal{E}_N(U_{d,\phi})=\underset{\theta}{\max}(\frac{1}{2}-\frac{1}{2^{2N-1}}\sqrt{A})\) is exactly same as the disentagling power, \(\mathcal{D}_N(U_{d,\phi})\), since only \(\cos \phi\)-term is involved in the expression. Hence, this completes the proof of Proposition~\ref{prop:alldiag_fullsep}.

\begin{figure}[ht]
\includegraphics[width=0.8\linewidth, height=6cm]{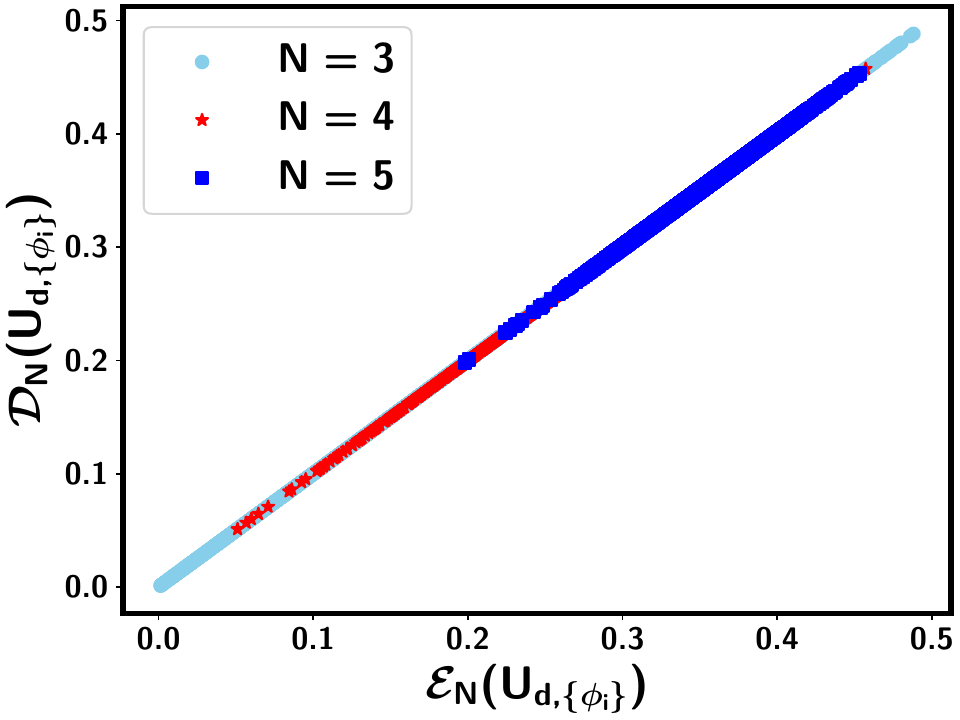}
\caption{\textbf{Equivalence between entangling and disentangling powers corresponding to diagonal unitaries.} Disentangling power, \(\mathcal{D}_N(U_{d,\{\phi_i\}})\) (vertical axis) versus entangling power, \(\mathcal{E}_N(U_{d,\{\phi_i\}})\) (horizontal axis) for randomly generated diagonal unitaries, where the optimization is performed over the fully separable states \(S_N\) for \(N = 3, 4, 5\). Here \(\{\phi_i\}\)s are chosen randomly from the Gaussian distribution with zero mean and unit standard deviation. The data points lie along the diagonal, indicating that \(\mathcal{E}_N(U_{d,\{\phi_i\}}) = \mathcal{D}_N(U_{d,\{\phi_i\}})\) holds for each value of \(N\). Both axes are dimensionless.
}

\label{fig:3_4_5qubit_diag_ggm} 
\end{figure}

\subsection{Entangling and disentangling powers coincide for random diagonal unitaries}
\label{sec:num_same}

Consider an eight-dimensional diagonal unitary operator, 
\( U_{d,\{\phi_i\}} = \mathrm{diag}(e^{i\phi_1}, e^{i\phi_2}, \ldots, e^{i\phi_8}) \), 
where each \(0\leq \phi_i\leq 2\pi\) is independently drawn from a Gaussian distribution, \(G(0,1)\) with vanishing mean  and unit standard deviation. We numerically simulate \(10^4\) randomly generated diagonal unitaries and observe that, for each case, the entangling 
\(\mathcal{E}_3(U_{d,\{\phi_i\}})\) and the disentangling 
\(\mathcal{D}_3(U_{d,\{\phi_i\}})\) powers match when the optimization (via numerical algorithm ISRES \cite{Runarsson2005}) is performed over the set of fully separable states \(S_3\), as illustrated in Fig.~\ref{fig:3_4_5qubit_diag_ggm}.
Notice that unlike the case of $U_{d,\phi}$ (a single-parameter diagonal unitaries), the optimal initial states in $\mathcal{E}_3(U_{d,\{\phi_i\}})$ and $\mathcal{D}_{3}(U_{d,\{\phi_i\}})$
depend on all $\theta_i$s which makes the optimization difficult although we notice that the optimal state are same in $\mathcal{E}_3(U_{d,\{\phi_i\}})$ and $\mathcal{D}_{3}(U_{d,\{\phi_i\}})$.

We have also numerically simulated $2^4$ - and $2^5$ - dimensional arbitrary diagonal unitaries by choosing  $\{\phi_i\}$s again from normal distributions and optimize over four- and five-qubit FS states, i.e., over $\{\theta_i\}_{i=1}^4$ and $\{\theta_i\}_{i=1}^5$ respectively. In both situations, we again find that local phases, $\{\xi_i\}$s, in the FS states do not contribute to optimization. The entangling powers, \(\mathcal{E}_4(U_{d,\{\phi_i\}})\) and \(\mathcal{E}_5(U_{d,\{\phi_i\}})\) coincide with their respective disentangling powers, \(\mathcal{D}_4(U_{d,\{\phi_i\}})\) and \(\mathcal{D}_5(U_{d,\{\phi_i\}})\) when optimized over the sets of fully separable states \(S_4\) and \(S_5\) (see Fig.~\ref{fig:3_4_5qubit_diag_ggm}). These numerical simulations possibly suggest that, like a single-parameter family of diagonal unitary operators, entangling and disentangling powers of arbitrary diagonal unitaries possibly involve only $\cos\phi_i$-like terms  and hence they match. It also indicates that such an equivalence may hold for arbitrary diagonal unitary operators acted on arbitrary number of qubits.

\begin{figure}
\includegraphics [width=\linewidth]{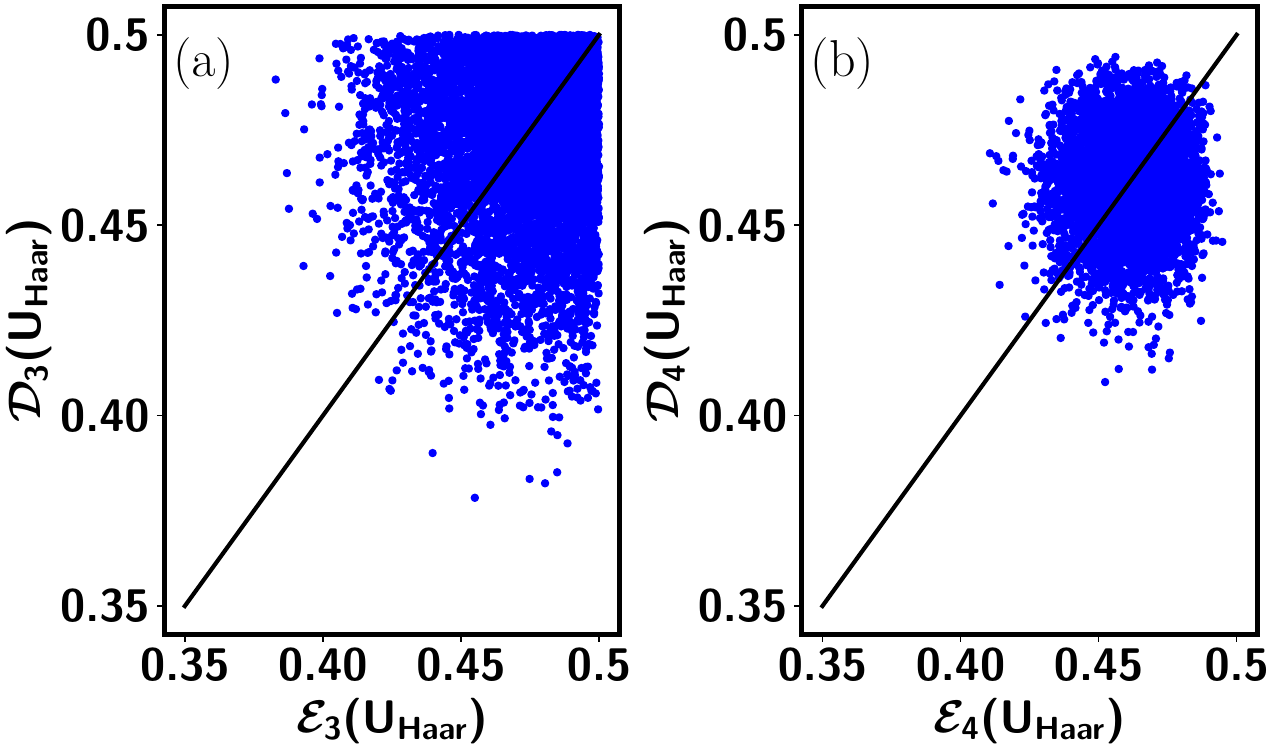}
\caption{ \textbf{Unequal multipartite entangling  and disentangling powers for random unitaries.} \((a)\) Disentangling power, \(\mathcal{D}_3(U_{\text{Haar}})\) (vertical axis) versus entangling power \(\mathcal{E}_3(U_{\text{Haar}})\) (horizontal axis) for Haar uniformly generated eight dimensional unitary operators. \((b)\) same as in (a) for  sixteen-dimensional unitaries. In both cases, unitaries are  generated Haar uniformly.  Solid line is the case when they match. Clearly, they differ in the plot unlike in Fig. \ref{fig:3_4_5qubit_diag_ggm}. Both axes are dimensionless.}
\label{fig:diff_ggm_rand} 
\end{figure}

\section{Inequivalence between multipartite entangling and disentangling powers with random unitaries}
\label{sec:inequ_app_c}

We observe that instead of constructing unitaries which act on pairwise sites, if one randomly samples \(10^4\) eight- and sixteen- dimensional unitary operators according to the Haar measure, a pronounced dissimilarity can again be observed between their multipartite entangling and disentangling capabilities. For depiction, we plot the disentangling power with respect to the entangling ones for a fixed unitary operator in Fig.~\ref{fig:diff_ggm_rand}. If \(\mathcal{E}_{N}(U_{\mathrm{Haar}})=\mathcal{D}_{N}(U_{\mathrm{Haar}})\), they all should have been on the diagonal, as observed in Fig.~\ref{fig:3_4_5qubit_diag_ggm} for diagonal unitaries. In this case, we clearly see that the points scatter around the diagonal, highlighting their differences.

\newpage
\bibliography{bib}

\end{document}